\documentclass[preprint,12pt,Dandan JIang, 2015]{elsarticle}

\usepackage{amssymb}

 \usepackage{lineno}
\usepackage{bm}
\usepackage{graphicx,epsfig,color,bm}

\usepackage{threeparttable}
\usepackage{longtable}
 \usepackage{supertabular}
 \usepackage{tabularx} 
  \usepackage{amsmath} \numberwithin{equation}{section} 
 \allowdisplaybreaks[3]

\newcommand{\bo}{{\bf 0}}
\newcommand{\bF}{{\bf F}}

\newcommand{\bone}{{\bf 1}}
\newcommand{\by}{{\bf y}}
\newcommand{\bX}{{\bf X}}
\newcommand{\bV}{{\bf V}}
\newcommand{\bx}{{\bf x}}
\newcommand{\bz}{{\bf z}}
\newcommand{\bxi}{{\bm \xi}}
\newcommand{\bbeta}{{\bm \eta}}
\newcommand{\bS}{{\bm \Sigma}}

\newcommand{\bmu}{{\bm \mu}}
\newcommand{\mb}{\mbox}
\newcommand{\md}{\mbox{d}}
\newcommand{\mS}{{\bf S}}

\newtheorem{theorem}{Theorem}[section]
\newtheorem{Lemma}{Lemma}[section]

\newtheorem{corollary}{Corollary}[section]

\newtheorem{proof}{Proof}[section]

\journal{Journal }

\begin{document}

\begin{frontmatter}



\title{Likelihood-based tests on linear  hypotheses of  large dimensional  mean vectors  with unequal  covariance matrices.}

\corref{mycorrespondingauthor}
\cortext[mycorrespondingauthor]{Corresponding author}

\author{Dandan Jiang\fnref{myfootnote}}
\fntext[myfootnote]{Supported by Project 11471140 from NSFC.}

\address{School of Mathematics, \\
Jilin University, \\
2699  QianJin Street, \\
Changchun {\rm 130012}, China.}
\ead{jiangdandan@jlu.edu.cn}

\begin{abstract}
This paper considers testing linear hypotheses of a set of  mean vectors   with unequal  covariance matrices in  large dimensional setting.  The problem of  testing the  hypothesis 
$H_0 : \sum_{i=1}^q  \beta_i  \bmu_i  =\bmu_0 $ for a given vector $\bmu_0$ is studied from the view of likelihood, which makes the  proposed tests   more powerful. 
We use  the  CLT for linear spectral statistics of  a large dimensional  $F$-matrix in \citet{Zheng}
 to establish the new test statistics in large dimensional framework,  so that the proposed tests can be  applicable for large dimensional  non-Gaussian variables in a wider range. Furthermore,  our new  tests provide more  optimal  empirical powers due to the likelihood-based statistics, meanwhile their empirical sizes are closer to the significant level.  Finally,  the simulation study is provided to compare the proposed tests with other high dimensional  mean vectors tests for evaluation of  their performances.

\end{abstract}

\begin{keyword}
Large dimensional data \sep Linear hypothesis  \sep  Mean vectors  tests\sep Random matrix theory

\MSC[2010] 62H15\sep  62H10

\end{keyword}

\end{frontmatter}


\section{Introduction} \label{Int}

Testing on mean vectors endures as an old, yet active research field with the applications of multiple comparisons, MANOVA  and classification. Continuing improvements on data acquisition techniques and the ease of access to high computation power  pose constant  challenges in applying the traditional statistical methods to these emerging data sets, 
because  they are established on the basis of fixed dimension $p$ as the sample size $n$ tends to infinity.  Within this context, more and more  attention is  paid to  find the efficient  testing methods for 
high dimension data  and much progress has been made in this respect. A special attention has been given to the linear hypothesis test of mean vectors,  which is an important  part of multivariate statistical analysis and   widely used in the  biology, finance and  etc.  
Suppose  $\bX_i=(\bx_{i1},\cdots,\bx_{in_i})',  i=1,\cdots,q $   to be the  independent   sample from 
$q$ population   with mean $\bmu_i$ and covariance matrix $\bS_i,  i=1,\cdots,q, $ respectively,  where  $\bS_i$'s are unequal $p\times p$ covariance matrix.
Consider the  test  hypothesis   
   \begin{equation}
   H_0 : \sum\limits_{i=1}^q  \beta_i  \bmu_i  =\bmu_0   \quad \mb{v.s.} \quad  H_1: \mb{not}~ H_0, \label{H1}
   \end{equation} 
where $\beta_1,\cdots, \beta_q$ are the given scalars and $\bmu_0$ is a  known vector. 
Of course, the  multivariate Behrens-Fisher problem and MANOVA are covered as  the special cases. A classical solution  was first proposed by \citet{Bennett}, see also in \citet{A03}, which was 
an extension of the methodology for two-sample case in \citet{Scheffe} to the multiple case. Then  many efforts have been devoted  to develop the solutions  in  the large dimensional data setting.   To be specific,  \citet{BS96} investigated  the two-sample case under the normal assumption and equal covariance matrices, and extended 
Hotelling's $T^2$ test  to the $p>N$ setting.  Motivated  by this work, \citet{Chen}  proposed a two-sample test for the  equality of the means of high dimensional data with unequal covariance matrices.  \citet{AY2011} derived  a nonparametric  test  for which the significant levels are not effected by the population distribution assumption.  Also,   \citet{Fujikoshi},  \citet{SriFuji}, \citet{Sri07}, \citet{Schott},   \citet{SriDu}, \citet{Sri09}, \citet{Sri13} and \citet{Hu2015}  were proposed for the test on the equality of high dimensional mean vectors. 
\citet{JSPI} focused on the testing linear hypothesis on the mean  vectors of normal populations with unequal covariance matrices when the dimensionality $p$ exceeds the sample size $n_i$. They proposed a new test procedure based on the Dempster trace criterion and showed its consistency  in high dimension setting.

Different from the previous works, we proposed  new tests  based on likelihood  for  the hypothesis (\ref{H1})   by  the  CLT for LSS of  a large dimensional  $F$-matrix in \citet{Zheng}.  The tests  in this work  were suitable for non-Gaussian variables in a wider range.  More important,  our proposed  tests  provided the more accurate sizes and  achieved much better performance on the empirical powers than other  high dimensional test methods, which had been  sustained by the simulation.    Finally,  the restricted condition was relaxed  to the finite $4$-th moment  compared  with \citet{Chen} and \citet{JSPI},  which made our tests  more applicable.

The remainder  of the article  is organized as follows. Section \ref{Pre}  gives   a quick review of the  linear hypothesis  test of mean vectors,  then  the  CLT for LSS of  a large dimensional  $F$-matrix in \citet{Zheng} is also provided in this part.  In Section \ref{New}, we propose  the  new  testing statistics in large dimensional setting  based on the classical  likelihood  test.  Simulation results are presented  to evaluate the performance   of  our test compared with other high dimensional mean vectors tests in Section \ref{Sim}. Finally, a conclusion is drawn in the  Section \ref{Con}, and the proofs and derivations are listed in the  \ref{App}


\section{Problem Description and Preliminary}  \label{Pre}
 In this section,  the problem of testing the linear hypothesis of mean vectors   is described in details.  As mentioned above,   a classical test  statistic   is proposed by Bennett (1951)  under Gaussian assumption, i.e.$\bX_i=(\bx_{i1},\cdots,\bx_{in_i})', i=1,\cdots,q $   be $n_i$ independent samples from $N_p(\bmu_i, \bS_i)$.
Without loss of generality,  assume $n_1$ is the  least  one if all the sample sizes 
$n_1,\cdots, n_q$ are different.  Set
\begin{equation}
\by_k=\beta_1\bx _{1 k}+\sum\limits_{i=2}^q \beta_i
\sqrt{\frac{n_1}{n_i}} \left(\bx _{i
k}-\frac{1}{n_1}\sum\limits_{l=1}^{n_1} \bx _{i l}+
\frac{1}{\sqrt{n_1n_i}}\sum\limits_{m=1}^{n_i}\bx _{i m}\right). \label{yx}
\end{equation}
where $k=1, \cdots, n_1$. It is simplified as $\by_k=\sum\limits_{i=1}^q\beta_i\bx _{i k}$,  if $n_1=n_2=\cdots=n_q$.
Then it is obviously that
\[\mb{E}\by_k=\sum_{i=1}^q\beta_i\bmu_i\]
and
\[
\mb{E}(\by_k-\mb{E}\by_k)(\by_l-\mb{E}\by_l)'=\delta_{kl}\left(\sum\limits_{i=1}^q\frac{\beta_i^2n_1}{n_i}\bS_i\right),
\]
where $\delta_{kl}$ is a  Kronecker's delta function. 
Meanwhile, make the denotations as  below
\[
\overline{\by}=\frac{1}{n_1}\sum\limits_{k=1}^{n_1}\by_k=\sum_{i=1}^q\beta_i\overline\bx
_i, \quad\overline\bx _i=\frac{1}{n_i}\sum\limits_{l=1}^{n_i}\bx _{il},
\]
and 
\begin{equation}
\mS=\frac{1}{n_1-1}\sum\limits_{k=1}^{n_1}(\by_k-\overline{\by})(\by_k-\overline{\by})'. \label{S}
\end{equation}
So the classical test statistic according to  Bennett (1951) is
\begin{equation}
T^2=n_1(\overline{\by}-\bmu_0)'\mS^{-1}(\overline{\by}-\bmu_0)\label{T2},
\end{equation}
which follows a $p$-dimensional  $T^2$ with $n_1-1$ freedom degree for any fixed $p$. 

However,  it is not the case when the dimension $p$ grows larger. Denote 
\[\sum_{i=1}^q\beta_i\bmu_i  \equiv \bmu, \quad   \sum\limits_{i=1}^q\frac{\beta_i^2n_1}{n_i}\bS_i \equiv \bS,\]
 then  we note that  $\by_1,\cdots,\by_{n_1}$ follows the distribution $N_p(\bmu,\bS)$ independently under the Gaussian assumption. Because 
 $\sqrt{n_1}(\overline\by-\bmu) \sim N_p(\bo,\bS)$,  then it is obtained that
 $n_1(\overline{\by}-\bmu_0)(\overline{\by}-\bmu_0)' \sim W_p(1, \bS) $ and
  $(n_1-1) {\mS}  \sim  W_p(n_1-1, \bS)$ under the null hypothesis,  and they are independent with each other,
 where $ W_p(f, \bS)$ means a $p$ dimensional Wishart distribution with freedom degree $f$ and parameter $\bS$.
Define the matrix
\begin{equation}
\bF=\frac{n_1(\overline{\by}-\bmu_0)(\overline{\by}-\bmu_0)'}{\mS},\label{F}
\end{equation}
according to  the definition of $F$-matrix, $\bF$ is an  $F$-matrix with  freedom degree $(1,n_1-1)$. Thus the classical 
$T^2$-test statistic  in (\ref{T2}) can be represented as 
\[
T^2=\mb{tr}\left(n_1(\overline{\by}-\bmu_0)(\overline{\by}-\bmu_0)'\mS^{-1}\right)=\mb{tr}(\bF)
\]
 Under the suitable 4-th moments constrains, 
 using the results on the limiting 
spectral distribution of $F$-matrix in eq. (4.4.1) in \citet{bookBS10}, it can be obtained with probability 1  
\begin{equation}
\frac{1}{p} T^2\!=\! \frac{1}{p}\! \sum\limits_{i=1}^p \!\lambda_i ^{\bF} 
\!\rightarrow \!  \int _{a}^b \!\frac {x(1\!-\!\gamma_2)\!\sqrt{(b-x)(x-a)}\md x}{2\pi x(\gamma _1+\gamma _2 x)}  \!=\!\frac{1}{1-\gamma_2}  \! \equiv\! d(\gamma_2)\! >\!1 
\label{limitd}
\end{equation}
where $\{ \lambda_i ^{\bF}, i=1,\cdots, p\}$ are the eigenvalues of the matrix $\bF$,  
$p/1=\gamma_p \rightarrow \gamma_1 \in (0,+\infty)$, 
$p/(n_1-1)=\gamma_{n_1} \rightarrow \gamma_2 \in (0,1)$ , $h=\sqrt{\gamma_1+\gamma _2-\gamma_1\gamma _2}$  and  
\[
a=\left( \frac  {1-h}{ 1-\gamma_2 }\right)^2,\qquad
b=\left( \frac  {1+h}{ 1-\gamma_2 }\right)^2.
\]
This result  is derived in the \ref{A1}.
As seen from above, it show that almost surely 
\[T^2= p\cdot d(\gamma_2)\] 
Thus, any test that assumes asymptotic $T$-square distribution of $T^2$ will  result in a serious error when $p$ grows higher and higher. 
Therefore, we intend to make some amendments to  the classical test by the  CLT of LSS ( linear spectral statistic ) of  large dimensional 
$F$-matrices, which is Theorem 3.2  in \citet{Zheng}.  In order to introduce it,
we first make clear  some  preliminary preparations.

 Let  $ \{\xi_{ki} \in \mathbb{C}, i, k = 1, 2, \cdots
\}$ and $\{\eta_{kj} \in \mathbb{C}, j, k = 1, 2, \cdots \}$ be either both real or both complex random arrays.   Write $\bxi_{\cdot i} = (\xi_{1i}, \xi_{2i},\cdots ,
\xi_{pi})'$ and $\bbeta_{\cdot j} = (\eta_{1j}, \eta_{2j}, \cdots,
\eta_{pj})'$. Also, for any positive integers $n_1, n_2$, 
 $ \bxi=(\bxi_{\cdot 1}, \cdots,\bxi_{\cdot n_1}) $ and $ \bbeta=(\bbeta_{\cdot
  1}, \cdots,\bbeta_{\cdot n_2}) $ can be thought as two independent samples of 
  a $p$-dimensional  observations of  size $ n_1 $and $ n_2 $, respectively.
Let $S_1$ and $S_2$ be the associated sample covariance matrices, 
 $i. e.$
\[
\mS_{1}={1\over{n_1}}\sum\limits_{i=1}^{n_1}\xi_{\cdot i}\xi_{\cdot i}^* \quad
\mbox{and}\quad
\mS_{2}={1\over{n_2}}\sum\limits_{j=1}^{n_2}\eta_{\cdot j}\eta_{\cdot
  j}^*.
\]
where $*$ stands for complex conjugate and transpose.
Then, the
following  so-called {\em F-matrix}  generalizes  the classical
Fisher-statistic to
the present $p$-dimensional case,
\begin{equation}
\bV_n=\mS_{1}\mS_{2}^{-1}\label{Fmatrix}
\end{equation}
where $n=(n_1, n_2)$ and  $n_2 > p$ is required to ensure that almost surely the matrix $\mS_2$ is 
invertible. 

 Let us also make some  assumptions as below:
 \begin{description}
\item[\emph{Assumption [A]}]  For any fixed $\epsilon_0>0$
 \begin{eqnarray*}
&& \frac{1}{n_1p}\sum\limits_{i=1}^p \sum\limits_{j=1}^{n_1}\mb{E}|\xi_{ij}|^4I(|\xi_{ij}|\geq \epsilon_0\sqrt{n_1}) \rightarrow 0; \\
&
& \frac{1}{n_2p}\sum\limits_{i=1}^p \sum\limits_{j=1}^{n_2}\mb{E}|\eta_{ij}|^4I(|\eta_{ij}|\geq \epsilon_0\sqrt{n_2} ) \rightarrow 0 
 \end{eqnarray*}
\item[\emph{Assumption [B]}]  The sample size $n_1, n_2$  and the dimension $p$ increase to infinity in such a large dimensional limiting  scheme  that 
\begin{equation}
  y_{n_1}=\frac{p}{n_1} \rightarrow y_1 \in (0, +\infty),  \quad
  y_{n_2}=\frac{p}{n_2} \rightarrow y_2
  \in (0, 1).\label{limitscheme}
\end{equation}
\end{description}

 Let $F_n^{\bV_n}$ denote the empirical spectral distribution(ESD) of the matrix $\bV_n$. Under the assumptions above, 
the ESD $F_n^{\bV_n}$ almost surely  converges  to the
  LSD (limiting spectral distribution) $F_{y_1, y_2}$ with the   density function represented as 
\begin{equation}
  \displaystyle{\ell(x)}=\left\{
  \begin{array}{ll} & \displaystyle{\frac {(1-y_{2})\sqrt{(b'-x)(x-a')}}{2\pi x(y_{1}+y_{2}x)},
      \quad  ~a' \leq x \leq b',}\\[6mm]
    &\displaystyle{0, \quad\quad \quad\quad \mbox{otherwise},}
  \end{array}
  \right.\label{LSDden}
\end{equation}
and has a point mass $1-\frac{1}{y_1} $ at the origin  if $y_1>1$, 
where $h'=\sqrt{y_{1}+y_{2}-y_{1}y_{2}} $
\[
a'=\left( \frac  {1-h'}{ 1-y_{2}}\right)^2,\qquad
b'=\left( \frac  {1+h'}{ 1-y_{2} }\right)^2.
\]
See  p.72 of \citet{bookBS10}.  We use $F_{y_{n_1}, y_{n_2}}$ to mark  
an analog representation of $F_{y_1, y_2}$ by substituting  the index
$y_{n_1}, y_{n_2}  $ for $y_1, y_2$.
Let $\mathcal{A}$   be a set of  functions $f_1,f_2,\cdots$, which are  analytic in an open region in the complex plane containing  the support of the continuous part of the LSD $F_{y_1, y_2}$ defined in (\ref{LSDden}).
A linear spectral statistic (LSS) of the random matrix ${\bV_n}$  is  expressed as
\[ \int f (x)\md F_n^{\bV_n}(x)=
\frac{1}{p}\sum\limits_{i=1}^p f(\lambda_i^{\bV_n}), ~~~~~~\quad f \in  \mathcal{A},
\]
where $\left(\lambda_i^{\bV_n}\right)$ are the real eigenvalues of  the $p
\times p$ square matrix $\bV_n$. Then based on  the empirical process
$G_n : = \{G_n(f)\}$ indexed by $\mathcal{A}$ ,
\begin{equation}
G_n(f)= p\cdot \int_{-\infty}^{+\infty} f(x)\left[F^{\bV_n}_n-
F_{y_{n_1}, y_{n_2}}\right] (\md x), ~~~~~~\quad f \in  \mathcal{A},\label{Gdef}
\end{equation}
the CLT for LSS of large dimensional $F$-matrices (Theorem 3.2 in \citet{Zheng})  is provided as following, which will play a fundamental role in next derivations.

Let
\[\kappa=
\left\{
\begin{array}{cc}
 2, & \mb{if the~}  \bxi,  \bbeta- \mb{variables are real,\quad\quad} \\
  1,& \mb{if the~}  \bxi, \bbeta- \mb{variables are complex.} 
\end{array}
\right.
\]
\begin{Lemma}   [Theorem 3.2 in \citet{Zheng}] \label{T2.1}
Assume that\\[-0.25in]
\begin{enumerate}
\item  Assumptions [A]-[B] are satisfied;
\item   For any positive integers $n_1, n_2$, 
 $ \bxi\!=\!(\bxi_{\cdot 1}, \!\cdots,\!\bxi_{\cdot n_1}) $  and  $ \bbeta \!=\! (\bbeta_{\cdot
  1}, \!\cdots,\!\bbeta_{\cdot n_2}) $ can be thought as two independent samples of 
  a $p$-dimensional  observations, where  $\bxi_{\cdot i} = (\xi_{1i}, \xi_{2i},\cdots ,
\xi_{pi})'$ and $\bbeta_{\cdot j} = (\eta_{1j}, \eta_{2j}, \cdots,
\eta_{pj})'$.  For all $i, j,k$,  $\mb{E}\xi_{ki}=\mb{E}\eta_{kj}=0$, 
$ \mb{E}{|\xi_{ki}|^2}=\mb{E}{|\eta_{kj}|^2}=\kappa-1$,  
 $\mb{E}{|\xi_{ki}|}^4 =\beta_x+\kappa-1< \infty$   and
$\mb{E}{|\eta_{kj}|}^4 =\beta_y+\kappa-1< \infty$,
where $\beta_x$ and  $\beta_y$ are contains  concerned with the 4-th moments.\\[-0.25in]
\end{enumerate}
Let  $ f_1, \cdots ,f_s \in \mathcal{A}$, then the
  random vector $ \left(G_n(f_1), \cdots,G_n(f_s)\right) $ weakly
  converges to a $s$-dimensional Gaussian vector  with the mean vector
 \begin{eqnarray}
&&\mu(f_j)=\frac{\kappa-1}{4\pi i} \! \oint f_j(z) \mb{d} 
\log \!\left(\!\frac{(1-y_2)m_0^2(z)+2m_0(z)\!+\!1\!-\!y_1}{(1-y_2)m_0^2(z)+2m_0(z)+1}\!\right) \label{mean1}\\
&&+\frac{\kappa-1}{4\pi i} \oint f_j(z) \md 
\log\left(1-y_2m_0^2(z)(1+m_0(z))^{-2}\right)\\ \label{mean2}
&&+\frac{\beta_x y_1 }{2 \pi i} \oint f_j(z)\left(1+m_0(z)\right)^{-3} \md m_0(z)\label{mean3}\\
&&+\frac{\beta_y  }{4 \pi i} \oint f_j(z)
\left(1-\frac{y_2m_0^2(z)}{(1+m_0(z))^{2}}\right) \md \log
\left(1-\frac{y_2m_0^2(z)}{(1+m_0(z))^{2}}\right) \label{mean4}
\end{eqnarray}
and covariance
function
\begin{eqnarray}
&&\upsilon\left(f_j,
f_\ell\right)=-\frac{\kappa}{4\pi^2}\oint \!\oint\frac{f_j(z_1)f_\ell(z_2)}{\left(m_0(z_1)-m_0(z_2)\right)^2}
\md m_0(z_1)\md m_0(z_2)    \label{var1}\\
&&-\!\frac{\beta_x y_1+\beta_y y_2}{4\pi^2}\!\oint \!\oint \!\frac{f_j(z_1)f_\ell(z_2)}{(1+m_0(z_1))^2(1+m_0(z_2))^2}
\! \md m_0(z_1)\! \md m_0(z_2)    \label{var2}
\end{eqnarray}
where  $ j,\ell \in \{1, \cdots,
s\}$,  $m_0(z)=\underline{m}_{y_2}(-\underline{m}(z))$.
Here $\underline{m}(z)$ is  the
Stieltjes Transform of ~ $\underline{F}_{y_1,y_2}\equiv (1-y_1)I_{[0,  \infty)}+y_1F_{y_1,y_2}$ and  $\underline{m}_{y_2}(z)$ is  the
Stieltjes Transform of ~ $\underline{F}_{y_2}\equiv (1-y_2)I_{[0,
    \infty)}+y_2F_{y_2}$, where $F_{y_2}$
 is the LSD of the matrix $\mS_2$.     The contours all contain the support of $F_{y_1,y_2}$ and  non overlapping  in both (\ref{var1})  and    (\ref{var2}).  
  \end{Lemma}
  The expression of the asymptotic mean and covariance in Lemma \ref{T2.1} is complicated to  figure it out.  So  further steps were given to help the evaluation of the asymptotic  mean and covariance in the Corollary~3.2 in  \citet{Zheng}.  However,  the result  provided in  the Corollary~3.2 in  \citet{Zheng} is not correct, I think it is a typo mistake. In order to obtain an accurate  and simplified  form for computing  the asymptotic mean and covariance, we reviewed 
the Corollary~3.2 in  \citet{Zheng}, and give the  corrected result   in the  following Lemma~\ref{Lem2}, which is proved in the \ref{A2}.

\begin{Lemma}\label{Lem2}
Under the assumptions of Lemma \ref{T2.1}, the asymptotic means and covariances of the limiting random
vector can be computed as follows
  \begin{eqnarray}
&&\!\mu(f_j)\!=\!\lim\limits_{\tau \downarrow
1}\!\frac{\kappa-1}{4\pi i} \!\oint_{|\xi|=1}\!
f_j \! \left(\!\frac{1\!+\!h'^2\!+\!2h'\mb{Re}(\xi)}{(1\!-\!y_2)^2}\!\right)\!\left[\!\frac{1}{\xi-{1\over \tau}}\!+\!\frac{1}{\xi+{1\over
\tau}}\!-\!\frac{2}{\xi+{y_2\over
{h' \tau}}} \!\right]\!\md\xi\label{E1}\nonumber \\
&&\\
&& +\frac{\beta_x\cdot y_1(1-y_2)^2}{2\pi i \cdot
h'^2}\oint_{|\xi|=1}f_j\left(\!\frac{1+h'^2+2h'\mb{Re}(\xi)}{(1-y_2)^2}\!\right)\frac{1}{(\xi+\frac{y_2}{h'})^3} \md \xi\label{E1betax}\\
&&+ \frac{\beta_y\cdot y_2(1-y_2)}{2 \pi i \cdot
h'}\oint_{|\xi|=1}f_j\left(\frac{1+h'^2+2h'\mb{Re}(\xi)}{(1-y_2)^2}\right)\frac{\xi +
\frac{1}{h'}}{(\xi+\frac{y_2}{h'})^3}\md\xi ,\label{E1betay}
\end{eqnarray}
where $ j=1, \cdots, s$, and covariance
function
\begin{eqnarray}
&&\upsilon\!\left(f_j,
f_\ell\right)\!=\!-\lim\limits_{ \tau \downarrow
    1}   \!\frac{\kappa}{4\pi^2}\!\oint_{|\xi_1\!|\!=\!1}
\!\oint_{|\xi_2\!|\!=\!1}\frac{f_j\!\left(\!\frac{1\!+\!h'^2\!+\!2h'\mb{Re}(\xi_1\!)}{(1-y_2)^2}\!\right)
f_\ell \!\left(\!\frac{1\!+\!h'^2\!+\!2h'\mb{Re}(\xi_2\!)}{(1-y_2)^2}\right)}{(\xi_1-\tau\xi_2)^2}
\!\md\xi_1\!\md\xi_2,\nonumber\\
&&\label{xicov1}\\
&& - \frac{(\beta_x y_1 \!+ \!\beta_y  y_2 \!)( \!1 \!- \!y_2 \!)^2}{4  \pi^2 \!h'^2}
 \!\oint_{|\xi_1\!| \!=\!1}  \! \frac{f_j  \!\left( \!\frac{ \!1 \!+ \!h'^2 \!+ \!2h'\mb{Re}(\xi_1)}{(1 \!- \!y_2)^2}\right)}{(\xi_1 \!+ \!\frac{y_2}{h'})^2}  \!\md\xi_1
 \!\oint_{|\xi_2\!| \!=\!1} \! \frac{f_j  \!\left( \!\frac{ \!1 \!+ \!h'^2 \!+ \!2h'\mb{Re}(\xi_2)}{(1 \!- \!y_2)^2}\right)}{(\xi_2 \!+ \!\frac{y_2}{h'})^2} \!\md\xi_2\nonumber\\
 &&\label{cov1betax}
\end{eqnarray}
where  $ j,\ell \in \{1, \cdots,
s\}$, "$\mb{Re}$" represents the real part of $\xi$ and $\tau \downarrow
1$ means that " $\tau$ approaches 1 from above'.
\end{Lemma}

 \section{ The Proposed Testing Statistics } \label{New}
 Based on  the CLT for the  $F$-matrices in  Lemma~\ref{T2.1}, a corrected  scaling for the classical test statistic is established. 
 Recall
that
\[T^2= \mb{tr}(\bF), \quad \bF=\frac{n_1(\overline{\by}-\bmu_0)(\overline{\by}-\bmu_0)'}{\mS}.\]
Under the null hypothesis $H_0$,  we have 
\[n_1(\overline{\by}-\bmu_0)(\overline{\by}-\bmu_0)' \sim W_p(1, \bS), \quad (n_1-1)\mS \sim W_p( n_1-1, \bS)\]   
  and they are independent with each other.
According to  the definition of $F$-matrix, 
standardization of the entries cannot effect on the values of   
 $\bF$, because both the numerator and denominator are  already centralized and have the same covariance 
parameter $\bS$.  Consequently,  $\bF$ is exactly distributed as the $F$-matrix
$\bV_n$ with  freedom degree $(1,n_1-1)$,  where in addition they have the same limiting spectral distributions. Thus,
our proposed   test statistic 
is given by Lemma \ref{T2.1} under the large dimensional  setting 
$p/1=\gamma_p \rightarrow \gamma_1 \in (0,+\infty) ~  \mb{and}~
p/(n_1-1)=\gamma_{n_1} \rightarrow \gamma_2 \in (0,1)$, which means that our method is valid for moderate high dimensionality. However, It  still works for ultra high dimensional data if there is a more larger sample size.

\begin{theorem} \label{CLT}
  Assuming that  the conditions of Lemma \ref{T2.1}  hold under
  $H_0$ in  (\ref{H1}), $T^2$ is 
  defined as in (\ref{T2}) and $
  f(x)=x$. Let $p/1=\gamma_p \rightarrow \gamma_1 \in (0,+\infty)$  and
$p/(n_1-1)=\gamma_{n_1} \rightarrow \gamma_2 \in (0,1)$.
  Then, under $H_0$ 
    \begin{equation}
 T_{ours}  =\upsilon(f)^{-\frac{1}{2}}\left[
    T^2-p \cdot d(\gamma_{n_1})-
      \mu(f)\right] \Rightarrow N \left( 0, 1\right).\label{ours}
  \end{equation}
  where $d(\gamma_{n_1})$ is derived in (\ref{limit}), and  $\mu(f),\upsilon(f)  $  are
  depicted as (\ref{testE}) and (\ref{testVar}), respectively.
\end{theorem}
\begin{proof}
According to the  definition in (\ref{T2}), we have
  \begin{eqnarray*}
   T^2  &=&\mb{tr}(\bF)  =\sum\limits_{i=1}^p \lambda^{\bF}_i=p \cdot \int x  dF_{n}^{\bF}(x)
  \end{eqnarray*}
  where $F_n^{\bF}(x)$ is the ESD of the matrix $\bF$ in  (\ref{F}).

  Since $\bF$ is exactly distributed as the $F$-matrix
$\bV_n$ with  freedom degree $n=(1,n_1-1)$,   $\bF$  has the same limiting spectral distribution 
  with the $F$-matrix $\bV_n$. Furthermore, the unbiased estimator of the covariance matrix of $\by_k,k=1,\cdots, n_1$ is adopted for  the denominator $\mS$ in $\bF$,  which is  the only  item subtracting sample mean. So  it is equivalent  to apply the CLT for LSS of large dimensional 
  $F$-matrix to either $\bF$ or $\bV_n$ with  freedom degree $n=(1,n_1-1)$.
  Then define $f(x)=x$  and 
   \begin{equation}
   G_n(f)=p \cdot \int f(x) d\left(F_n^{\bF}(x)-F_{\gamma_{p},
      \gamma_{n_1}}(x)\right),\label{ESD-pLSD}
       \end{equation}
        where $F_{\gamma_{p},
      \gamma_{n_1}} $   is analogous  to LSD of the matrix $\bF$,  which has
  a density in (\ref{LSDden}) but with $\gamma_{p},
      \gamma_{n_1} $ instead of $y_k,
  k=1,2.$, respectively. Consequently,
 $F_{\gamma_{p},
      \gamma_{n_1}}(f)=\int f(x)dF_{\gamma_{p},
      \gamma_{n_1}}(x) $ is exactly analogous to   the $d(\gamma_2)$ calculated in (\ref{limitd}) by substituting $\gamma_{n_1}$ for $\gamma_2$, i.e.
      \begin{equation}
      F_{\gamma_{p},
      \gamma_{n_1}}(f)=\int f(x)dF_{\gamma_{p},
      \gamma_{n_1}}(x) =\frac{1}{1-\gamma_{n_1}} \equiv d(\gamma_{n_1})
      \label{limit}
\end{equation}

    By Lemma~\ref{T2.1}, $ G_n(f)$  weakly converges to a Gaussian
  vector with mean

  \begin{equation}
    \mu(f)= \frac{\gamma_2}{(1-\gamma_2)^2}+\frac{\beta_y\gamma_2}{1-\gamma_2}\label{testE}
  \end{equation}
  and variance
  \begin{equation}
    \upsilon(f)=\frac{2h^2}{(1-\gamma_2)^4}+\frac{\beta_x\gamma_1+\beta_y\gamma_2}{(1-\gamma_2)^2}\label{testVar}
  \end{equation}
 where $ h=\sqrt{\gamma_1+\gamma_2-\gamma_1\gamma_2}$,  $\beta_x$ and $\beta_y$ here are the Kurtosis  of the standardized 
  $\bar \by$ and $\by_i$,respectively, which can be calculated from (\ref{yx}). But it is complicated and we simulated them in the simulation study.  (\ref{testE}) and (\ref{testVar}) are calculated by Lemma~\ref{T2.1} in the   \ref{A3} and \ref{A4}. 
From
  \begin{eqnarray}
   T^2&=&p \cdot \int f(x)dF_n^{\bF}(x)\nonumber\\
   &=& p \cdot \int f(x) d\left(F_n^{\bF}(x)-F_{\gamma_{p},
      \gamma_{n_1}}(x)\right) +p \cdot F_{\gamma_{p},
      \gamma_{n_1}}(f),\\
      &=& G_n(f)+p \cdot d( \gamma_{n_1})\nonumber
  \end{eqnarray}
 we get
  \begin{eqnarray}
   G_n(f)= T^2-p\cdot d( \gamma_{n_1})    &\Rightarrow & N\left(\mu(f), \upsilon(f)\right).\label{2connec}
  \end{eqnarray}
Therefore,   
  $$
  T_{ours}=\upsilon(f)^{-\frac{1}{2}}\left[ T^2-p
    \cdot  d(\gamma_{n_1})- \mu(f)\right] \Rightarrow N \left( 0,
  1\right).
  $$
\end{proof}

The  test statistic we proposed for testing
(\ref{H1})
is  based on the likelihood ratio  test  statistic $T^2$ the and its asymptotic
distribution is  derived in the theorem  above.
However, it is worth noticing that in the above proof, we used the Gaussian
assumption for entry variables to fit $F$-matrix definition, but
Lemma~\ref{T2.1} does not need this Gaussian assumption. Therefore,
 the asymptotic distribution for
 in Theorem~\ref{CLT} could be applied  more
generally to non-Gaussian  variables. The simulations could certainly make out a case for this point of view.

Next, we consider  some special cases of  the test hypothesis (\ref{H1}), and derive their test statistics and asymptotic distributions in some corollaries. 
First,  we focus on the testing the equality of two population mean vectors with unequal covariance matrices. It is  also well known as the multivariate Behrens-Fisher problem.That is 
\begin{equation}
H_0: \bmu_1=\bmu_2 \quad  \mb{v.s.} \quad H_1:  \bmu_1 \neq \bmu_2,
\end{equation}
 which is a special case of  the hypothesis (\ref{H1}) with $q=2, \beta_1=1, \beta_2=-1$  and  $\bmu_0=\bo$.  
Then define 
\[
\by_k=\bx _{1 k}-\sqrt{\frac{n_1}{n_2}} \bx _{2
k}+\frac{1}{\sqrt{n_1n_2}}\sum\limits_{l=1}^{n_1} \bx _{2 l}-
\frac{1}{n_2}\sum\limits_{m=1}^{n_2}\bx _{2 m}.
\]
and $n_1< n_2$ without loss of generality.
Thus, we also have 
$\{\by_k, k=1,\cdots, n_1\}$  are independent  and 
\[\bmu \equiv \mb{E}\by_k=\bmu_1-\bmu_2\]
and
\[
\bS\equiv \mb{E}(\by_k-\mb{E}\by_k)(\by_l-\mb{E}\by_l)'=\delta_{kl}\left(\bS_1+\frac{n_1}{n_2}\bS_2\right),
\]
where $\delta_{kl}$ is a  Kronecker's delta function.  So it is equivalent to test 

\[H_0: \bmu=\bo \quad  \mb{v.s.} \quad H_1:  \bmu \neq  \bo\]
and the classical test statistic is
\begin{equation}
T_{BF}=n_1\overline{\by}'\mS^{-1}\overline{\by},
\end{equation}
where $\mS$ is defined in (\ref{S}) with $q=2$. 
Applying the Theorem~\ref{CLT}, we have the following corollary:
\begin{corollary}
For testing $H_0 : \bmu_1=\bmu_2 $ with unequal covariance matrix $\bS_i, i=1,2.$, under the assumption of  Theorem \ref{CLT},
we have   the conclusion of Theorem \ref{CLT} still holds,
only with the test statistic $T^2$ in (\ref{ours}) is revised by $T_{BF}$.
\end{corollary}

For more simplicity,  we assume all of the variables have  the common covariance matrix, that is $\bS_1=\cdots=\bS_q=\bS$. Then for testing the hypothesis (\ref{H1}) with the common covariance matrix assumption, we set
\[
\bar\by=\sum\limits_{i=1}^q \beta_i\bar\bx _i
\quad 
\mS=\frac{1}{\sum\limits_{i=1}^q n_i-q}\sum\limits_{i=1}^q\sum\limits_{k=1}^{n_i}(\bx_{ik}-\overline\bx_i)(\bx_{ik}-\overline\bx_i)'
\]
where $\overline\bx _i=\displaystyle\frac{1}{n_i}\sum\limits_{k=1}^{n_i}\bx _{ik}$.
So the classical likelihood test statistic  is
\begin{equation}
T_{M}=\sum\limits_{i=1}^q\frac{\beta_i^2}{n_i}(\overline{\by}-\bmu_0)'\mS^{-1}(\overline{\by}-\bmu_0),
\end{equation}
Define the matrix
\begin{equation}
\bF_1=\frac{\sum_{i=1}^q\frac{\beta_i^2}{n_i}(\overline{\by}-\bmu_0)(\overline{\by}-\bmu_0)'}{\mS},
\end{equation}
according to  the definition of $F$-matrix, $\bF_1$ is an  $F$-matrix with  freedom degree $(1,\sum\limits_{i=1}^q n_i-q)$. Thus 
the test statistic  $T_M$ can be written as 
\[
T_M=\mb{tr}\left(\sum\limits_{i=1}^q\frac{\beta_i^2}{n_i}(\overline{\by}-\bmu_0)(\overline{\by}-\bmu_0)'\mS^{-1}\right)=\mb{tr}(\bF_1)
\]
Applying the Theorem  \ref{CLT}, the corresponding corollary are  given as below.

\begin{corollary}
For testing $H_0 : \sum\limits_{i=1}^q\beta_i\bmu_i=\bmu_0$ with common covariance matrix $\bS$,  assume that 
$p/1=\gamma_p \rightarrow \gamma_1 \in (0,+\infty)$  and
$p/(\sum\limits_{i=1}^q n_i-q)=\gamma_{n} \rightarrow \gamma_2 \in (0,1)$
and other conditions  of  Theorem \ref{CLT} still hold, then 
we have   the conclusion of Theorem \ref{CLT}  
only with the test statistic $T^2$ in (\ref{ours}) is revised by
 $T_M$.
\end{corollary}

For  the test on the equality of  two mean vectors  with common  covariance matrix,  we have 
\[
T_D =\frac{n_1n_2}{n_1+n_2}\overline{\by}'\mS^{-1}\overline{\by}=\mb{tr}(\bF_2)
\]
where 
\[\overline{\by}=\bar\bx_1-\bar\bx_2, \quad  \mS=\frac{1}{n_1+n_2-2}\sum\limits_{i=1}^2\sum\limits_{k=1}^{n_i}(\bx_{ik}-\overline\bx_i)(\bx_{ik}-\overline\bx_i)'
\]
and 
\[\bF_2=\frac{n_1n_2}{n_1+n_2}\overline{\by}\overline{\by}'\mS^{-1}\]
is satisfied for the definition of $F$-matrix with freedom degree $(1, n_1+n_2-2)$.
Applying the Theorem  \ref{CLT}, the corresponding corollary is  given as below.

\begin{corollary}
For testing $H_0 : \bmu_1=\bmu_2 $ with equal covariance matrix $\bS$,   assume that 
$p/1=\gamma_p \rightarrow \gamma_1 \in (0,+\infty)$  and
$p/(n_2+n_2-2)=\gamma_{n} \rightarrow \gamma_2 \in (0,1)$
and other conditions  of  Theorem \ref{CLT} still hold, then 
we have   the conclusion of Theorem \ref{CLT}  
only with the test statistic $T^2$ in (\ref{ours}) is revised by
 $T_D$.
\end{corollary}

 
\section{Simulation Study}\label{Sim}

 In this section, simulations  are  conducted to evaluate the test statistics    that we proposed based on likelihood $T^2$ test statistic. Two hypotheses $  H_{0a} : \sum_{i=1}^3  \beta_i  \bmu_i  =\bo $  and  $H_{0b}: \bmu_1=\bmu_2$ with unequal covariance matrices are investigated  without loss of generality.  
We also  present the corresponding simulation results of other  tests as a comparison, for examples tests in \citet{JSPI}(TNT) for $H_{0a}$  and  the tests  in \citet{JSPI}(TNT)  and \citet{Chen} (CQT) for $H_{0b}$. The samples are  generate from the model 
$$
\bx_{ij}=\Gamma_i\bz_{ij}+\bmu_i, \quad  i=1,\cdots,q,  j= 1,\cdots, n_i
$$
where $\bz_{ij}=(z_{ij1}, \cdots, z_{ijp})'$ and $\{z_{ijk}, k=1,\cdots,p \}$ are independently distributed  as one of the following distribution assumptions:
$$
\mb { (i) }  N(0,1); \quad \quad \quad    \mb { (ii) } Gamma(4,0.5)-2;
$$

For the covariance matrix $\bS_i, i=1,2,3$,  the following  cases concerned with the dimension $p$ are taken into account  by 
\begin{eqnarray}
&&\bS_i=\Gamma_i^2=W_i \Phi_i W_i\\
&&W_i= \mb{diag}(w_{i1},\cdots, w_{ip}), w_{ij}=2\times i+(p-j+1)/p\\
&&\Phi_i =\left(\phi^{(i)}_{jk}\right), \phi^{(i)}_{jj}=1,  \phi^{(i)}_{jk}=(-1)^{(j+k)}(0.2\times i)^{|j-k|^{0.1}},  j \neq k,
\end{eqnarray}
which is cited from \citet{Hu2015}.
The suitable mean vectors $\bmu_i, i=1,2,3$ are chosen  for different  hypotheses. 
 
 First, for two-sample problem $H_{0a}: \bmu_1=\bmu_2$,  the null hypothesis is assumed to be $\bmu_1=\bmu_2=\bo$ without  loss of generality.  Denote  $\Delta \bmu=(\epsilon\sqrt{2\log(p)}\cdot\bone'_{[p^{v_0}]},\bo'_{p-[p^{v_0}]})'$, where 
$\bone_p$ represents a vector with that all elements  are 1, $[\cdot]$ denotes the integer truncation function and $\epsilon, v_0$  are varying constants.  
For the alternative hypothesis,  the sparse model similar to the one in \citet{Chen} is applied,  which describes $
\bx_{ij}=\Gamma_i\bz_{ij}+\bmu_i,  i=1,2,  j= 1,\cdots, n_i$
 and $\bmu_1=\bo$, $\bmu_2=\Delta \bmu$

 Secondly, we consider  the three groups testing problem $  H_{0b} : \sum_{i=1}^3  \beta_i  \bmu_i  =\bo $. Under the null hypothesis, we choose two cases of $\beta_i, i=1,2,3$.  One is $\beta_1=\beta_2=\beta_3=1$, and the corresponding   mean vectors  are generally  selected as $\bmu_1=1, \bmu_2=1,  \bmu_3=-2$. The other one is $\beta_1=\beta_2=-\frac{1}{2}, \beta_3=1$, and the corresponding   mean vectors  are given as  $\bmu_1=1, \bmu_2=3,  \bmu_3=2$ without loss of generality. The alternative hypotheses are designed  that  the $\bmu_3$ is the value under the null hypothesis added  $\Delta\bmu$ described as above, while $\bmu_1$ and $\bmu_2$ remain unchanged. 

  For each set of the scenarios, we report both empirical Type I errors and powers with 10,000 replications at $\alpha=0.05$ significance  level. Different pair values of $p, n_1, n_2, n_3$ are selected, and  $\epsilon$ is varying from 0 to  0.9 or 1 to show the empirical sizes and powers. The mean parameter is supposed to be unknown and substituted by the sample mean during the calculations. Simulation results of empirical Type I errors and powers for the three group tests  are listed in the Table  \ref{tab:1} and Table  \ref{tab:2}.  Simulation results of empirical Type I errors and powers for the two-sample test  are  represented in
 Table \ref{tab:3}.

   \begin{table}[htbp]
  \centering\caption{ Empirical sizes and powers of the comparative tests for  $H_0 : \sum\limits_{i=1}^3  \beta_i  \bmu_i  =\bo  $  with $\beta_1=\beta_2=\beta_3=1$ at$~\alpha=0.05$ significance level  for normal and gamma random vectors with 10,000  replications.  The alternative hypothesis is    $\bmu_3=-(\bmu_1+\bmu_2)+\bmu_{\Delta}$, $\bmu_{\Delta}=(\epsilon\sqrt{2\log(p)}\cdot\bone'_{[p^{v_0}]},\bo'_{p-[p^{v_0}]})'$ 
  \label{tab:1} }
   \begin{tabularx}{14cm}{ccXXXX}   
\hline
  {$(p, n_1, n_2, n_3)$  }         &              & \multicolumn{2}{c}{$(40, 90, 100, 100)$  }                  &\multicolumn{2}{c}{$(40, 180, 200, 200)$}                          \\
     &&Ours & TNT&Ours & TNT\\
     \hline
        &    & \multicolumn{2}{c}{\underline{\quad\quad\quad  $v_0$=0.4\quad \quad \quad }}                    &\multicolumn{2}{c}{\underline{\quad\quad\quad  $v_0$=0.2\quad \quad \quad}}    \\
  Normal  &  $\epsilon=0$ (size) &0.0647           &0.0697         &0.0614          &0.0707  \\ 
               &  0.2                           &0.1056           &0.0772         &0.1059          &0.0721  \\ 
               &  0.4                           &0.2568           &0.0872         &0.3042          &0.0855  \\ 
               &  0.6                           &0.5795           &0.1235         &0.6888          &0.1181  \\ 
               &  0.8                           &0.8739           &0.2055         &0.9496          &0.2038  \\ 
               &  1                              &0.9854           &0.4134         &0.9986          &0.4014  \\ 
                &    & \multicolumn{2}{c}{\underline{\quad\quad\quad  $v_0$=0.5\quad \quad \quad }}                    &\multicolumn{2}{c}{\underline{\quad\quad\quad  $v_0$=0.3\quad \quad \quad}}    \\
 Gamma &  $\epsilon=0$ (size)  &0.0758           &0.0726         &0.0641          &0.0690  \\ 
               &  0.2                           &0.1022           &0.0748         &0.0960          &0.0738  \\ 
               &  0.4                           &0.2071           &0.0838         &0.2400          &0.0865  \\ 
               &  0.6                           &0.4491           &0.1194         &0.5571          &0.1106  \\ 
               &  0.8                           &0.7804           &0.1609         &0.8848          &0.1660  \\ 
               &  1                              &0.9588           &0.2896         &0.9912          &0.2780  \\ 
  \end{tabularx} 
  \begin{tabularx}{14cm}{ccXXXX}   
\hline
  {$(p, n_1, n_2, n_3)$  }         &              & \multicolumn{2}{c}{$(80, 180, 200, 200)$  }                  &\multicolumn{2}{c}{$(120, 180, 200, 200)$}                          \\
     &&Ours & TNT &Ours & TNT\\
      \hline
         &    & \multicolumn{2}{c}{\underline{\quad\quad\quad  $v_0$=0.2\quad \quad \quad }}                    &\multicolumn{2}{c}{\underline{\quad\quad\quad  $v_0$=0.3\quad \quad \quad}}    \\
  Normal  &  $\epsilon=0$ (size) &0.0643           &0.0714         &0.0678          &0.0705  \\ 
               &  0.2                           &0.0906           &0.0750         &0.1085          &0.0757  \\ 
               &  0.4                           &0.2093           &0.0781         &0.2869          &0.0944  \\ 
               &  0.6                           &0.4895           &0.0985         &0.6701          &0.1143  \\ 
               &  0.8                           &0.8224           &0.1329         &0.9490          &0.1798  \\ 
               &  1                              &0.9773           &0.1890         &0.9980          &0.3558 \\ 
                 &    & \multicolumn{2}{c}{\underline{\quad\quad\quad  $v_0$=0.4\quad \quad \quad }}                    &\multicolumn{2}{c}{\underline{\quad\quad\quad  $v_0$=0.4\quad \quad \quad}}    \\
 Gamma &  $\epsilon=0$ (size)  &0.0662           &0.0683         &0.0743          &0.0732  \\ 
               &  0.2                           &0.0986           &0.0738         &0.0980          &0.0758  \\ 
               &  0.4                           &0.2887           &0.0847         &0.2392          &0.0848 \\ 
               &  0.6                           &0.6743           &0.1127         &0.5462          &0.1096  \\ 
               &  0.8                           &0.9572           &0.1776         &0.8808          &0.1615  \\ 
               &  1                              &0.9993           &0.3008         &0.9928          &0.2566  \\ 
 \hline
  \end{tabularx} 
  \end{table}

   \begin{table}[htbp]
  \centering\caption{ Empirical sizes and powers of the comparative tests for  $H_0 : \sum\limits_{i=1}^3  \beta_i  \bmu_i  =\bo  $  with $\beta_1=\beta_2=-\frac{1}{2}, \beta_3=1$ at$~\alpha=0.05$ significance level  for normal and gamma random vectors with 10,000  replications.  The alternative hypothesis is    $\bmu_3=\frac{1}{2}(\bmu_1+\bmu_2)+\bmu_{\Delta}$, $\bmu_{\Delta}=(\epsilon\sqrt{2\log(p)}\cdot\bone'_{[p^{v_0}]},\bo'_{p-[p^{v_0}]})'$ 
  \label{tab:2} }
  \begin{tabularx}{14cm}{ccXXXX}   
\hline
  {$(p, n_1, n_2, n_3)$  }         &              & \multicolumn{2}{c}{$(40, 90, 100, 100)$  }                  &\multicolumn{2}{c}{$(40, 180, 200, 200)$}                          \\
     &&Ours & TNT &Ours & TNT\\
      \hline
         &    & \multicolumn{2}{c}{\underline{\quad\quad\quad  $v_0$=0.3\quad \quad \quad }}                    &\multicolumn{2}{c}{\underline{\quad\quad\quad  $v_0$=0.1\quad \quad \quad}}    \\
  Normal  &  $\epsilon=0$ (size) &0.0691           &0.0705         &0.0670          &0.0692  \\ 
               &  0.2                           &0.1107           &0.0742         &0.0965          &0.0713  \\ 
               &  0.4                           &0.3162           &0.0895         &0.2204          &0.0845  \\ 
               &  0.6                           &0.6827           &0.1194         &0.4947          &0.0961  \\ 
               &  0.8                           &0.9433           &0.1935         &0.8139          &0.1299  \\ 
               &  1                              &0.9972           &0.3499         &0.9686          &0.1958  \\ 
                 &    & \multicolumn{2}{c}{\underline{\quad\quad\quad  $v_0$=0.1\quad \quad \quad }}                    &\multicolumn{2}{c}{\underline{\quad\quad\quad  $v_0$=0.1\quad \quad \quad}}    \\
 Gamma &  $\epsilon=0$ (size)  &0.0730           &0.0715         &0.0672          &0.0729  \\ 
               &  0.2                           &0.1093           &0.0723         &0.1093          &0.0723  \\ 
               &  0.4                           &0.2748           &0.0894         &0.6591          &0.0967  \\ 
               &  0.6                           &0.6045           &0.1067         &0.9886          &0.1465 \\ 
               &  0.8                           &0.9083           &0.1448         &1                   &0.2561  \\ 
               &  1                              &0.9933           &0.1947         &1                   &0.5925  \\ 
  \end{tabularx} 
  \begin{tabularx}{14cm}{ccXXXX}   
\hline
  {$(p, n_1, n_2, n_3)$  }         &              & \multicolumn{2}{c}{$(80, 180, 200, 200)$  }                  &\multicolumn{2}{c}{$(120, 180, 200, 200)$}                          \\
     &&Ours & TNT &Ours & TNT\\
     \hline
        &    & \multicolumn{2}{c}{\underline{\quad\quad\quad  $v_0$=0.2\quad \quad \quad }}                    &\multicolumn{2}{c}{\underline{\quad\quad\quad  $v_0$=0.2\quad \quad \quad}}    \\

  Normal  &  $\epsilon=0$ (size) &0.0626           &0.0676         &0.0682          &0.0678  \\ 
               &  0.2                           &0.1173           &0.0772         &0.1081          &0.0706  \\ 
               &  0.4                           &0.3596           &0.0799         &0.2471          &0.0814  \\ 
               &  0.6                           &0.7925           &0.1062         &0.5629          &0.0942  \\ 
               &  0.8                           &0.9848           &0.1526         &0.8832          &0.1167 \\ 
               &  1                              &0.9999           &0.2524         &0.9900          &0.1717  \\
                    &    & \multicolumn{2}{c}{\underline{\quad\quad\quad  $v_0$=0.1\quad \quad \quad }}                    &\multicolumn{2}{c}{\underline{\quad\quad\quad  $v_0$=0.1\quad \quad \quad}}    \\ 
 Gamma &  $\epsilon=0$ (size)  &0.0658           &0.0684         &0.0726          &0.0712  \\ 
               &  0.2                           &0.1324           &0.0746         &0.1107          &0.0709  \\ 
               &  0.4                           &0.4737           &0.0812         &0.3199          &0.0789  \\ 
               &  0.6                           &0.9214           &0.1070         &0.7183          &0.0978  \\ 
               &  0.8                           &0.9994           &0.1457         &0.9655          &0.1225 \\ 
               &  1                              &1                    &0.2375         &0.9993          &0.1725 \\ 
\hline
  \end{tabularx} 
  \end{table}

  \begin{table}[htbp]
  \centering\caption{ Empirical sizes and powers of the comparative tests for  $H_0 :  \bmu_1 =\bmu_2  $  at$~\alpha=0.05$ significance level  for normal and gamma random vectors with 10,000  replications.  The alternative hypothesis is    $\bmu_1=\bo$, $\bmu_2=(\epsilon\sqrt{2\log(p)}\cdot\bone'_{[p^{v_0}]},\bo'_{p-[p^{v_0}]})'$ 
  \label{tab:3} }
  \begin{tabularx}{14cm}{cXXXcXXX}   
\hline
     & \multicolumn{3}{c}{\underline{\quad\quad\quad  $v_0$=0.2\quad \quad \quad }}      &              &\multicolumn{3}{c}{\underline{\quad\quad\quad  $v_0$=0.1\quad \quad \quad}}    \\
  {$(p, n_1, n_2)$  }                       & \multicolumn{3}{c}{$(40, 90, 100)$  }       &            &\multicolumn{3}{c}{$(40, 180, 200)$}                          \\
     &Ours & TNT &CQT& &Ours & TNT&CQT \\
     \hline
        & \multicolumn{3}{c}{Normal }            &      &\multicolumn{3}{c}{Normal}  \\
                  $\epsilon=0$ (size) &0.0678           &0.0689         &0.0679      & $\epsilon=0$  &0.0645          &0.0663       &0.0677 \\ 
                 0.2                           &0.1068           &0.0805         &0.0696       &0.2                  &0.0953          &0.0782       &0.0701  \\ 
                 0.4                           &0.2403           &0.1179         &0.0705       &0.4                   &0.2444          &0.1195       &0.0750   \\ 
                 0.6                           &0.5222           &0.2344         &0.0725       &0.6                   &0.5588          &0.2344       &0.0770 \\ 
                 0.8                           &0.8323           &0.5470         &0.0776       &0.8                   &0.8715          &0.5360         &0.0790  \\ 
                 1                              &0.9748           &0.8863         &0.0926       &0.9                   &0.9549          &0.7267         &0.0850  \\ 
                   & \multicolumn{3}{c}{Gamma}               &   &\multicolumn{3}{c}{Gamma}  \\
                $\epsilon=0$ (size)   &0.0681           &0.0711         &0.0710        &$\epsilon=0$&0.0618          &0.0670          &0.0687 \\ 
                 0.2                           &0.1030           &0.0786         &0.0717        &0.2                  &0.0938          &0.0788       &0.0702   \\ 
                 0.4                           &0.2313           &0.1161         &0.0709        &0.4                  &0.2312          &0.1217       &0.0752   \\ 
                 0.6                           &0.5200           &0.2382         &0.0732        &0.6                   &0.5532          &0.2396       &0.0801         \\ 
                 0.8                           &0.8309           &0.5633         &0.0821        &0.8                 &0.8704          &0.5371         &0.0897  \\ 
                 1                              &0.9697           &0.8915         &0.0943       &0.9                   &0.9531          &0.7279         &0.0982  \\ 
  \end{tabularx} 
  \begin{tabularx}{14cm}{cXXXcXXX}   
\hline
     & \multicolumn{3}{c}{\underline{\quad\quad\quad  $v_0$=0.1\quad \quad \quad }}             &      &\multicolumn{3}{c}{\underline{\quad\quad\quad  $v_0$=0.2\quad \quad \quad}}    \\
  {$(p, n_1, n_2)$  }                       & \multicolumn{3}{c}{$(80, 180, 200)$  }          &       &\multicolumn{3}{c}{$(120, 180, 200)$}                          \\
&Ours & TNT &CQT& &Ours & TNT&CQT \\
  \hline
         & \multicolumn{3}{c}{Normal }                &  &\multicolumn{3}{c}{Normal}  \\
                  $\epsilon=0$ (size) &0.0661           &0.0706          &0.0680    &$\epsilon=0$      &0.0638          &0.0707         &0.0641  \\ 
                 0.2                           &0.0872           &0.0713          &0.0743     &0.2                     &0.1041          &0.0746         &0.0698  \\ 
                 0.4                           &0.1769           &0.0978          &0.0790     &0.4                     &0.2687          &0.1125         &0.0704  \\ 
                 0.6                           &0.4007           &0.1510          &0.0802      &0.6                    &0.6170          &0.2288         &0.0773   \\ 
                 0.8                           &0.6993           &0.2966          &0.0803      &0.8                    &0.9166          &0.6044         &0.0814   \\ 
                 1                              &0.9280           &0.5999          &0.0820       &0.9                   &0.9781           &0.8514        &0.0823 \\ 
                     & \multicolumn{3}{c}{Gamma}          &        &\multicolumn{3}{c}{Gamma}  \\
                 $\epsilon=0$ (size)  &0.0606           &0.0722          &0.0694          &$\epsilon=0$&0.0647          &0.0674         &0.0666\\ 
                 0.2                           &0.0828           &0.0724          &0.0752          &0.2               &0.1007          &0.0798         &0.0705 \\ 
                 0.4                           &0.1720           &0.0959          &0.0786          &0.4                &0.2609          &0.1131         &0.0699   \\ 
                 0.6                           &0.3856           &0.1486          &0.0795          &0.6               &0.6093          &0.2282         &0.0781   \\ 
                 0.8                           &0.6970           &0.2865          &0.0798          &0.8              &0.9114          &0.6149         &0.0815   \\ 
                 1                              &0.9279           &0.6010          &0.0818         & 0.9             &0.9741          &0.8499         &0.0842  \\ 
\hline
  \end{tabularx} 
  \end{table}

For three groups tests, as seen from the Table~\ref{tab:1} and Table~\ref{tab:2},  the advantages of our proposed tests   compared with the tests in \citet{JSPI}(TNT) for $H_{0a}$ are listed in two aspects. First,  almost all  of the empirical Type I errors of our proposed test are  around the nominal size 5\%, which 
 are better than that of TNT.  Although,  the empirical size of
  the proposed test is slightly higher  for the case of $p=40, n_1=90, n_2=n_3=100 $ under the Gamma assumption, it  can be accepted and understood due to both asymptotic and nonparametric. Further more, it decreases with the increasing dimension  $p$ and sample size $n_1$.

Secondly,  it is obvious that our proposed tests give  a much better performance on the empirical powers, which uniformly dominates that of the TNT over the entire range. For examples,  under the Normal assumption  with 
$\beta_1=\beta_2=-\frac1{2}, \beta_3=1$ in Table \ref{tab:2}, our empirical powers is 96.86\% closing to 1, while the one of TNT is only around 20\% for the case of $p=40, n_1=180, n_2=n_3=200$ and $\epsilon=1, v_0=0.1$. When  the dimension   increases to  $p=120$,  the empirical power of our proposed 
test rises up to 99\%, but the one of TNT remains under  20\%.

For  two-samples test,  we compared our  test with  the ones in \citet{JSPI} (TNT)  and \citet{Chen}  (CQT) for $H_{0b}$  together.  As seen from the Table~\ref{tab:3},  it was the same thing for the comparison to TNT. First,  all the empirical sizes of our proposed test  are   one upon that  of  TNT, and our empirical powers  grows up to 1 rapidly,  which are superior to that of TNT. 
Then  let  us make a comparative analysis between CQT and our test.
 It can be easily found that  the empirical sizes of the CQT are slightly higher than that of our proposed test    for  almost all  cases.   Further,   CQT behaves even worse on the empirical powers,  like  $p=120, n_1=180, n_2=n_3=200$ and $\epsilon=0.9, v_0=0.2$ under the Normal assumption,  the  empirical power of  the CQT remains under  10\%  when our  empirical power increases to 1.

Finnally, It must be pointed that the proposed test cannot be use for the ultra high dimension $p> n_1$. On one hand,  the condition of 
$p< n_1$ is requested to guarantee   the inverse of  sample covariance  matrix of $\by_i$. On the other hand, the limiting variance 
$\nu(f)$ is related with  $\gamma_1$. If the dimension $p$ is large enough, it will make the proposed test unstable.


\section{Conclusion}\label{Con}

In this paper,  the  new testing statistics based on likelihood  were proposed for the linear  hypotheses tests of  the large dimensional mean  vectors  with unequal  covariance matrices.  By using   the  CLT for LSS of  a large dimensional  $F$-matrix in \citet{Zheng}, we guaranteed  that the tests proposed were feasible for the non-Gaussian variables in a wider range. Furthermore,  our test  methods provided the more  optimal powers due to the likelihood based statistics, meanwhile the empirical sizes were closer to the significant level. However, it is limited by the constrain  $p< n_1$, which is requested  for the existence of  the inverse of  sample covariance  matrix. For future works, maybe we can extend this work to other forms of test statistics by  large dimensional  spectral analysis in random matrix theory, and make it more powerful and applicable  for    different situations.


\section*{Acknowledgement}
The author thanks the reviewers for their helpful comments and suggestions to make an improvement of this article. 
This research was supported by the National Natural Science Foundation of China 11471140.



\appendix
\section{Proofs}\label{App}

\subsection{ Derivation of $d(\gamma_2)$ in (\ref{limitd}). \label{A1} }

Let $F_{\gamma_1,\gamma_2}(x) $ be the LSD of the matrix $\bF$,  
and denote  the integral 
\[ F_{\gamma_1,\gamma_2}(f) =\int _{a}^b \frac {x\cdot (1-\gamma_2)\sqrt{(b-x)(x-a)}}{2\pi x(\gamma _1+\gamma _2 x)} \md x,\]
where $f(x)=x$ and 
\[a=\left( \frac  {1-\sqrt{\gamma_1+\gamma _2-\gamma_1\gamma _2}}{ 1-\gamma_2 }\right)^2,\qquad
b=\left( \frac  {1+\sqrt{\gamma_1+\gamma _2-\gamma_1\gamma _2}}{ 1-\gamma_2 }\right)^2.
\]

Set  $
x=\frac{\left(1+h^2+2h\cos\theta\right)}{(1-\gamma_2)^{2}},  0 <
\theta < \pi$, where 
$h=\sqrt{\gamma_1+\gamma_2-\gamma_1\gamma_2}$. Then
\begin{eqnarray}
\sqrt{(b-x)(x-a)}= \frac{2h\sin\theta}{(1-\gamma_2)^2},
\quad&\quad
\md x=-\displaystyle\frac{2h\sin\theta}{(1-\gamma_2)^2} \md \theta;\nonumber\\
x=\frac{\left|1+he^{i\theta}\right|^2}{(1-\gamma_2)^2},\quad&\quad
\displaystyle{\gamma_1+\gamma_2 x}=\displaystyle{\frac{\left|h+\gamma_2 e^{i\theta}\right|^2}
{(1-\gamma_2)^2}}.\nonumber
\end{eqnarray}
So we have
\begin{eqnarray*}
 &&\int _{a}^b \frac {x\cdot (1-\gamma_2)\sqrt{(b-x)(x-a)}}{2\pi x(\gamma _1+\gamma _2 x)} \md x \\&=&
 -\frac{2}{\pi (1-\gamma_2)}\int^{0}_{\pi}
  \frac{h^2\sin^2\theta}{\left|h+\gamma_2e^{i\theta}\right|^2} \md \theta\\
  &=&\frac{1}{\pi (1-\gamma_2)}\int_{0}^{2\pi}
  \frac{h^2\sin^2\theta}{\left|h+\gamma_2e^{i\theta}\right|^2} \md \theta\\
  &=&-\frac{1}{4\pi  i(1-\gamma_2)}\oint_{|\xi|=1}
  \frac{h^2(\xi-\xi^{-1})^2}{ 
      \left|h+\gamma_2\xi\right|^2 \xi} \md\xi\\
 &=&-\frac{h}{4\pi  i(1-\gamma_2)\gamma_2}\oint_{|\xi|=1} 
 \frac{(\xi^2-1)^2}{
      (\xi+\frac{h}{\gamma_2})(\xi+\frac{\gamma_2}{h} )\xi^2 } \md\xi\end{eqnarray*}
There are two poles inside the unit circle: 0, $-\frac{\gamma_2}{h}$.
Their corresponding residues are
\begin{eqnarray}
\mb{  Res}(0) &=& \frac{-h^2-\gamma_2^2}{\gamma_2 h} ,\nonumber\\
   \mb{ Res}(-\frac{\gamma_2}{h})&=&
   \frac{h^2-\gamma_2^2}{\gamma_2 h}.\nonumber
\end{eqnarray}

Therefore
\begin{eqnarray*}
F_{\gamma_1,\gamma_2}(f) &=&-\frac{h}{4\pi  i(1-\gamma_2)\gamma_2}
\cdot 2\pi i \left(
\mb{  Res}(0) +\mb{ Res}(-\frac{\gamma_2}{h})\right) \\
&=&\frac{1}{1-\gamma_2} \equiv d(\gamma_2).
\end{eqnarray*}

Similarly,  
 $d(\gamma_{n_1})$  in Theorem \ref{CLT}  is exactly analogous to   the $d(\gamma_2)$  by substituting $\gamma_{n_1}$ for $\gamma_2$, i.e.
      \[F_{\gamma_{p},
      \gamma_{n_1}}(f)=\int f(x)dF_{\gamma_{p},
      \gamma_{n_1}}(x) =\frac{1}{1-\gamma_{n_1}} \equiv d(\gamma_{n_1})
\]
 where $F_{\gamma_{p},
      \gamma_{n_1}} $   is analogous  to LSD of the matrix $\bF$,  which has
  a density in (\ref{LSDden}) but with $\gamma_{p},
      \gamma_{n_1} $ instead of $y_k,
  k=1,2.$, respectively.

\subsection{  Derivations of the corrected   Corollary~3.2 in  \citet{Zheng} . \label{A2}}

Because it is difficult to apply Lemma \ref{T2.1} directly, which has the complex form of  the asymptotic mean and covariance. 
So  the Corollary~3.2 in  \citet{Zheng} was proposed to help the evaluation of the asymptotic  mean and covariance.  However,  the result of  the Corollary~3.2 in  \citet{Zheng} is not correct,  I think it is a typo mistake. In order to obtain the accurate  and simplified  form for computing  the asymptotic mean and covariance, we reviewed 
it and gave some derivations and calculations as below.

First, make clear some notations:
\begin{itemize}
\item  $m(z)$ is  the
Stieltjes Transform of the LSD $F_{y_1,y_2}$, where  $F_{y_1,y_2}$  is the LSD of the $F$-matrix $\bV_n$.
Define 
\begin{equation}
\underline{m}(z)=-\frac{1-y_1}{z}+y_1m(z),\label{mz}
\end{equation}
then
$\underline{m}(z)$ is  the
Stieltjes Transform of ~ $\underline{F}_{y_1,y_2}\equiv (1-y_1)I_{[0,  \infty)}+y_1F_{y_1,y_2}$,  which  has an inverse equation as
\begin{equation}
z=-\frac{1}{\underline{m}(z)}+y_1\int \frac{\md F_{y_2}(x)}{x+\underline{m}(z)},
\label{inversemz}
\end{equation}
where $F_{y_2}$
 is the LSD of the matrix $\mS_2$. 
 \item Denote $m_{y_2}(z)$ is the Stieltjes Transform of  $F_{y_2}$, consequently 
\begin{equation}
\underline{m}_{y_2}(z)=-\frac{1-y_2}{z}+y_2m_{y_2}(z)\label{my2z}
\end{equation}
 is  the
Stieltjes Transform of ~ $\underline{F}_{y_2}\equiv (1-y_2)I_{[0,
    \infty)}+y_2F_{y_2}$,  which has an inverse 
    \begin{equation}
z=-\frac{1}{\underline{m}_{y_2}(z)}+\frac{y_2}{1+\underline{m}_{y_2}(z)},
\label{inversemy2z}
\end{equation}
Therefore, equation (\ref{inversemz}) can be written as
\begin{eqnarray}
z&=&-\frac{1}{\underline{m}(z)}+y_1\int \frac{\md F_{y_2}(x)}{x+\underline{m}(z)},\nonumber\\
&=&-\frac{1}{\underline{m}(z)}+y_1\cdot m_{y_2}\big(-\underline{m}(z)\big)\nonumber\\
&=&-\frac{1}{\underline{m}(z)}-\frac{y_1(1-y_2)}{y_2 \underline{m}(z)}+ \frac{y_1}{y_2} \underline{m}_{y_2}\big(-\underline{m}(z)\big)\nonumber\\
&=&-\frac{y_1+y_2-y_1y_2}{y_2\underline{m}(z)} + \frac{y_1}{y_2} \underline{m}_{y_2}\big(-\underline{m}(z)\big)
\label{zmzm0}
\end{eqnarray}

\item Let $m_0(z)=\underline{m}_{y_2}\big(-\underline{m}(z)\big)$, for simplicity denote it as $m_0$ if no confusion.  By the inverse equation (\ref{inversemy2z}), we have 
\begin{equation}
\underline{m}(z)=\frac{(1-y_2)\left(m_0+1/(1-y_2)\right)}{m_0(1+m_0)}. \label{mzm0}
\end{equation}
Combine equation (\ref{zmzm0}) and (\ref{mzm0}),  the relationship between $z$ and $m_0$ is obtained 
\begin{equation}
z=-\frac{m_0\left(m_0+1-y_1\right)}{(1-y_2)\left(m_0+1/(1-y_2)\right)}. \label{zm0}
\end{equation}

\end{itemize}
Since  the contour enclosed the supporting set of the LSD $F_{y_1,y_2}(x)$  of the $F$-matrix $\bV_n$,  which contains the interval 
\[\left[a'=\frac{(1-h')^2}{(1-y_2)^2},  b'=\frac{(1+h')^2}{(1-y_2)^2}\right]\]
if  $y_1 \leq 1$ , where $h'=\sqrt{y_1+y_2-y_1y_2}$. 
 When $y_1>1$,  the contour should enclose the whole  support $\{0\}  \cup [ a', b']$, because the $F_{y_1,y_2}$ has a positive mass at the origin at this time.  However,  due to the exact separation theorem in \citet{BS99}, for large enough $p$ and $n$, the discrete mass at the origin  will coincide  with that of $F_{y_1,y_2}$. So we can restrict 
the integral on the contours  only enclosed the continuous part of the LSD  $F_{y_1,y_2}$.
Therefore,  solve the real roots of the equation (\ref{zm0}) at two points $a',b'$,  we obtain 
\[m_0(a')=-\frac{1-h'}{1-y_2}, \quad  m_0(b')=-\frac{1+h'}{1-y_2}\]
It is obviously  that   when $z$ runs in the positive  direction around  the interval $[a', b']$, $m_0(z)$ runs in the same direction around the interval $[-\frac{1-h'}{1-y_2}, -\frac{1+h'}{1-y_2}]$.  So define $m_0(z)=-\frac{1+h'\tau\xi}{1-y_2}$, where $\tau>1$ but very close to 1, and  $|\xi|=1$.  By  (\ref{zm0}), 
\begin{equation}
z=\frac{1+h'^2+h'\tau^{-1}\bar \xi+h' \tau \xi }{(1-y_2)^2}.\label{zxi}
\end{equation}
 Further, 
 \[\underline{m}'(z)=-\frac{(1-y_2)m^2_0+2m_0+1}{m^2_0(1+m_0)^2}  \cdot \md m_0 
 =\frac{(1-y_2)^2}{h' \tau}\cdot \frac{\left(\xi+\frac{\sqrt{y_2}}{h' \tau}\right )\left(\xi-\frac{\sqrt{y_2}}{h' \tau}\right ) }{\left(\xi+\frac{y_2}{h' \tau}\right )^2  \left(\xi+\frac{1}{h' \tau}\right )^2} \md \xi\]
 and 
 \[\underline{m}(z)=-\frac{(1-y_2)^2}{h' \tau} \frac{\xi}{\left(\xi+\frac{y_2}{h' \tau}\right )  \left(\xi+\frac{1}{h' \tau}\right )}.\]
 
 Put these results into the expressions of the asymptotic mean  and covariance in Lemma~\ref{T2.1}. According to the definition (\ref{zxi}), when $z$ anticlockwise runs along the unit circle, $z$ anticlockwise runs  around a contour closely enclosed  the interval $[a', b']$ when $\tau $ is closed to 1. Thus, letting  $\tau \downarrow 1 $, we have 
 
  \begin{eqnarray}
  &&\mu(f_j)=\frac{\kappa-1}{4\pi i} \! \oint f_j(z) \mb{d} 
\log \!\left(\!\frac{(1-y_2)m_0^2(z)+2m_0(z)\!+\!1\!-\!y_1}{(1-y_2)m_0^2(z)+2m_0(z)+1}\!\right)  \nonumber\\
&&+\frac{\kappa-1}{4\pi i} \oint f_j(z) \md 
\log\left(1-y_2m_0^2(z)(1+m_0(z))^{-2}\right) \nonumber\\ 
&&+\frac{\beta_x y_1 }{2 \pi i} \oint f_j(z)\left(1+m_0(z)\right)^{-3} \md m_0(z) \nonumber\\
&&+\frac{\beta_y  }{4 \pi i} \oint f_j(z)
\left(1-\frac{y_2m_0^2(z)}{(1+m_0(z))^{2}}\right) \md \log
\left(1-\frac{y_2m_0^2(z)}{(1+m_0(z))^{2}}\right) \nonumber\\[1mm]
&&=\!\lim\limits_{\tau \downarrow
1}\!\frac{\kappa-1}{4\pi i} \!\oint_{|\xi|=1}\!
f_j \! \left(\!\frac{1\!+\!h'^2\!+\!2h'\mb{Re}(\xi)}{(1\!-\!y_2)^2}\!\right)\!\left[\!\frac{1}{\xi-{1\over \tau}}\!+\!\frac{1}{\xi+{1\over
\tau}}\!-\!\frac{2}{\xi+{y_2\over
{h' \tau}}} \!\right]\!\md\xi\label{E1}\nonumber \\
&& +\frac{\beta_x\cdot y_1(1-y_2)^2}{2\pi i \cdot
h'^2}\oint_{|\xi|=1}f_j\left(\!\frac{1+h'^2+2h'\mb{Re}(\xi)}{(1-y_2)^2}\!\right)\frac{1}{(\xi+\frac{y_2}{h'})^3} \md \xi\nonumber\\
&&+ \frac{\beta_y\cdot y_2(1-y_2)}{2 \pi i \cdot
h'}\oint_{|\xi|=1}f_j\left(\frac{1+h'^2+2h'\mb{Re}(\xi)}{(1-y_2)^2}\right)\frac{\xi +
\frac{1}{h'}}{(\xi+\frac{y_2}{h'})^3}\md\xi ,\nonumber
\end{eqnarray}
where $ j=1, \cdots, s$, and covariance
function
\begin{eqnarray*}
&&\upsilon\left(f_j,
f_\ell\right)=-\frac{\kappa}{4\pi^2}\oint \!\oint\frac{f_j(z_1)f_\ell(z_2)}{\left(m_0(z_1)-m_0(z_2)\right)^2}
\md m_0(z_1)\md m_0(z_2)    \nonumber\\
&&-\!\frac{\beta_x y_1+\beta_y y_2}{4\pi^2}\!\oint \!\oint \!\frac{f_j(z_1)f_\ell(z_2)}{(1+m_0(z_1))^2(1+m_0(z_2))^2}
\! \md m_0(z_1)\! \md m_0(z_2)   \nonumber\\
&&\upsilon\!\left(f_j,
f_\ell\right)\!=\!-\lim\limits_{ \tau \downarrow
    1}   \!\frac{\kappa}{4\pi^2}\!\oint_{|\xi_1\!|\!=\!1}
\!\oint_{|\xi_2\!|\!=\!1}\frac{f_j\!\left(\!\frac{1\!+\!h'^2\!+\!2h'\mb{Re}(\xi_1\!)}{(1-y_2)^2}\!\right)
f_\ell \!\left(\!\frac{1\!+\!h'^2\!+\!2h'\mb{Re}(\xi_2\!)}{(1-y_2)^2}\right)}{(\xi_1-\tau\xi_2)^2}
\!\md\xi_1\!\md\xi_2,\\
&& - \frac{(\beta_x y_1 \!+ \!\beta_y  y_2 \!)( \!1 \!- \!y_2 \!)^2}{4  \pi^2 \!h'^2}
 \!\oint_{|\xi_1\!| \!=\!1}  \! \frac{f_j  \!\left( \!\frac{ \!1 \!+ \!h'^2 \!+ \!2h'\mb{Re}(\xi_1)}{(1 \!- \!y_2)^2}\right)}{(\xi_1 \!+ \!\frac{y_2}{h'})^2}  \!\md\xi_1
 \!\oint_{|\xi_2\!| \!=\!1} \! \frac{f_j  \!\left( \!\frac{ \!1 \!+ \!h'^2 \!+ \!2h'\mb{Re}(\xi_2)}{(1 \!- \!y_2)^2}\right)}{(\xi_2 \!+ \!\frac{y_2}{h'})^2} \!\md\xi_2\nonumber
\end{eqnarray*}
where  $ j,\ell \in \{1, \cdots,
s\}$, "$\mb{Re}$" represents the real part of $\xi$ and $\tau \downarrow
1$ means that " $\tau$ approaches 1 from above'.

\subsection{   Calculation of  $\mu (f)$ in (\ref{testE}).\label{A3}}

For  the function $f(x)=x$,  the computation of  $\mu (f)$ is divided  into three parts. 
Still use the denotation $h=\sqrt{\gamma_1+\gamma_2-\gamma_1\gamma_2}$, then the first part is 

\begin{eqnarray*}
I_1&=&\lim\limits_{\tau \downarrow
1}\frac{\kappa-1}{4\pi i}\oint_{|\xi|=1}
f\left(\frac{1+h^2+2h\mb{Re}(\xi)}{(1-\gamma_2)^2}\right)\left[\frac{1}{\xi-{1\over \tau}}+\frac{1}{\xi+{1\over
\tau}}-\frac{2}{\xi+{\gamma_2\over
{h \tau}}}\right] \md\xi\\
&=&\lim\limits_{\tau \downarrow
1} \frac{\kappa-1}{4\pi i}\oint_{|\xi|=1} \frac{\left|1+h\xi\right|^2}{(1-\gamma_2)^2}
\left(\frac{1}{\xi-{1\over
\tau}}+\frac{1}{\xi+{1\over \tau}}-\frac{2}{\xi+{\gamma_2\over {h\tau}}}
\right) \md \xi\\
&=&\lim\limits_{\tau \downarrow
1} \frac{(\kappa-1)h}{4\pi i (1-\gamma_2)^2}\oint_{|\xi|=1} \frac{\left(\xi+{1 \over h}\right)\left(\xi+h\right)}{\xi}
\left(\frac{1}{\xi-{1\over
\tau}}+\frac{1}{\xi+{1\over \tau}}-\frac{2}{\xi+{\gamma_2\over {h\tau}}}
\right) \md \xi\\
&=& \lim\limits_{\tau \downarrow
1} \frac{(\kappa-1)h}{4\pi i (1-\gamma_2)^2} \cdot 2 \pi i \left[ \mb{Res}(0)+\mb{Res}({1\over \tau })+\mb{Res}(-{1\over \tau} )-2\mb{Res}(-\frac{\gamma_2}{h\tau})\right]\nonumber\\
&=&\frac{(\kappa-1)h}{4\pi i (1-\gamma_2)^2} \cdot 2 \pi i \left[-\frac{2h}{\gamma_2}+2+\frac{1+h^2}{h}-2+\frac{1+h^2}{h}+ \frac{2(1-\gamma_2)^2\gamma_1}{h\gamma_2}\right]\nonumber\\
&=&\frac{\gamma_2}{(1-\gamma_2)^2}
\end{eqnarray*}

The second  part is 

\begin{eqnarray*}
I_2&=&\frac{\beta_1\cdot \gamma_1(1-\gamma_2)^2}{2\pi i \cdot
h^2}\oint_{|\xi|=1}f\left(\frac{1+h^2+2h\mb{Re}(\xi)}{(1-\gamma_2)^2}\right)\frac{1}{(\xi+\frac{\gamma_2}{h})^3} \md \xi\\
&=&\frac{\beta_1\cdot \gamma_1(1-\gamma_2)^2}{2\pi i \cdot
h^2}\oint_{|\xi|=1} \frac{\left|1+h\xi\right|^2}{(1-\gamma_2)^2}
\frac{1}{(\xi+\frac{\gamma_2}{h})^3} \md \xi\\
&=&\frac{\beta_1\cdot \gamma_1}{2\pi i 
h}\oint_{|\xi|=1}  \frac{\left(\xi+{1 \over h}\right)\left(\xi+h\right)}{\xi}
\frac{1}{(\xi+\frac{\gamma_2}{h})^3} \md \xi\\
&=&\frac{\beta_1\cdot \gamma_1}{2\pi i 
h} \cdot 2 \pi i \left[ \mb{Res}(0)+\mb{Res}(-\frac{\gamma_2}{h})\right]\nonumber\\
&=&\frac{\beta_1\cdot \gamma_1}{2\pi i 
h} \cdot 2 \pi i \left[ \frac{h^3}{y^3_2}-\frac{h^3}{y^3_2}\right]\nonumber\\
&=&0
\end{eqnarray*}

The third  part is 

\begin{eqnarray*}
I_3&=& \frac{\beta_2\cdot \gamma_2(1-\gamma_2)}{2 \pi i \cdot
h}\oint_{|\xi|=1}f\left(\frac{1+h^2+2h\mb{Re}(\xi)}{(1-\gamma_2)^2}\right)\frac{\xi +
\frac{1}{h}}{(\xi+\frac{\gamma_2}{h})^3}\md\xi \\
&=& \frac{\beta_2\cdot \gamma_2(1-\gamma_2)}{2 \pi i \cdot
h}\oint_{|\xi|=1} \frac{\left|1+h\xi\right|^2}{(1-\gamma_2)^2}
\frac{\xi +\frac{1}{h}}{(\xi+\frac{\gamma_2}{h})^3}\md\xi \\
&=&\frac{\beta_2 \gamma_2}{2 \pi i (1-\gamma_2)
}\oint_{|\xi|=1}  \frac{\left(\xi+{1 \over h}\right)\left(\xi+h\right)}{\xi}
\frac{\xi +\frac{1}{h}}{(\xi+\frac{\gamma_2}{h})^3}\md\xi \\
&=&\frac{\beta_2 \gamma_2}{2 \pi i (1-\gamma_2)
} \cdot 2 \pi i \left[ \mb{Res}(0)+\mb{Res}(-\frac{\gamma_2}{h})\right]\nonumber\\
&=&\frac{\beta_2 \gamma_2}{2 \pi i (1-\gamma_2)
} \cdot 2 \pi i  \left[ \frac{h^2}{y^3_2}1-\frac{h^2}{y^3_2}\right]\nonumber\\
&=&\frac{\beta_2 \gamma_2}{ 1-\gamma_2}
\end{eqnarray*}

Finally, 
\[ \mu(f)= \frac{\gamma_2}{(1-\gamma_2)^2}+\frac{\beta_2\gamma_2}{1-\gamma_2}\]

\subsection{  Calculation of  $\upsilon (f)$ in (\ref{testVar}). \label{A4}}

 The computation of  $\upsilon (f)$ in (\ref{testVar}) is divided  into two parts. For the first part
\begin{eqnarray*}
&&- \lim\limits_{ \tau \downarrow
    1}   \frac{\kappa}{4\pi^2}\oint_{|\xi_2|=1}
\oint_{|\xi_1|=1}\frac{f\left(\frac{1+h^2+2h\mb{Re}(\xi_1)}{(1-\gamma_2)^2}\right)f_\ell\left(\frac{1+h^2+2h\mb{Re}(\xi_2)}{(1-\gamma_2)^2}\right)}{(\xi_1-\tau\xi_2)^2}
\md\xi_1\md\xi_2
\end{eqnarray*}
the following integral is computed firstly
\begin{eqnarray*}
&&\oint_{|\xi_1|=1}f\left(\frac{1+h^2+2h\mb{Re}(\xi_1)}{(1-\gamma_2)^2}\right)\frac 1{(\xi_1-\tau\xi_2)^2}\md\xi_1\\
&=&\oint_{|\xi_1|=1}\frac{\left|1+h\xi\right|^2}{(1-\gamma_2)^2}\frac 1{(\xi_1-\tau\xi_2)^2}\md\xi_1\\
&=&\frac{h}{(1-\gamma_2)^2} \oint_{|\xi_1|=1}\frac{\left(\xi+{1 \over h}\right)\left(\xi+h\right)}{\xi (\xi_1-\tau\xi_2)^2}\md\xi_1\\
&=&\frac{2\pi i h}{(1-\gamma_2)^2} \cdot \frac{1}{\tau^2\xi_2^2}
\end{eqnarray*}
Then 
we obtained 
\begin{eqnarray*}
&&- \lim\limits_{ \tau \downarrow
    1}   \frac{\kappa}{4\pi^2}\oint_{|\xi_2|=1}
\oint_{|\xi_1|=1}\frac{f\left(\frac{1+h^2+2h\mb{Re}(\xi_1)}{(1-\gamma_2)^2}\right)f_\ell\left(\frac{1+h^2+2h\mb{Re}(\xi_2)}{(1-\gamma_2)^2}\right)}{(\xi_1-\tau\xi_2)^2}
\md\xi_1\md\xi_2\\
&=&- \lim\limits_{ \tau \downarrow
    1}   \frac{\kappa}{4\pi^2}\cdot \frac{2\pi i h}{(1-\gamma_2)^2}\oint_{|\xi_2|=1}\frac{\left|1+h\xi_2\right|^2}{(1-\gamma_2)^2} \frac{1}{\tau^2\xi_2^2}\md\xi_2\\
    &=&-\lim\limits_{ \tau \downarrow
    1}    \frac{\kappa h^2}{2\pi i (1-\gamma_2)^4\tau^2}\oint_{|\xi_2|=1} \frac{\left(\xi_2+{1 \over h}\right)\left(\xi_2+h\right)}{\xi_2^3}\md\xi_2\\
    &=& \frac{\kappa h^2}{ (1-\gamma_2)^4}
\end{eqnarray*}

For the second part,
we calculate 
\begin{eqnarray*}
&&\oint_{|\xi_1|=1}\frac{f\left(\frac{1+h^2+2h\mb{Re}(\xi_1)}{(1-\gamma_2)^2}\right)}{(\xi_1+\frac{\gamma_2}{h})^2}\md\xi_1\\
&=& \frac{h}{(1-\gamma_2)^2} \oint_{|\xi_1|=1} \frac{\left(\xi_1+{1 \over h}\right)\left(\xi_1+h\right)}{\xi_1(\xi_1+\frac{\gamma_2}{h})^2}\md\xi_1\\
&=& \frac{2\pi ih}{(1-\gamma_2)^2}  \left(\frac{h^2}{\gamma_2^2} +\frac{\gamma_2^2-h^2}{\gamma_2^2}\right)\\
&=& \frac{2\pi ih}{(1-\gamma_2)^2}
\end{eqnarray*}
Then
the second part is 
 \begin{eqnarray*}
&& -\frac{(\!\beta_1 \!\gamma_1\!+\!\beta_2\!\gamma_2\!)(1\!-\!\gamma_2\!)^2}{4\pi^2h^2}
\!\oint_{|\xi_1\!|\!=\!1}\!\frac{f_j \!\left(\!\frac{1\!+\!h^2\!+\!2h\mb{Re}(\xi_1\!)}{(1-\gamma_2)^2}\right)}{(\xi_1+\frac{\gamma_2}{h})^2}\!\md\xi_1
\!\oint_{|\xi_2\!|\!=\!1}\!\frac{f_j \!\left(\!\frac{1\!+\!h^2\!+\!2h\mb{Re}(\xi_2\!)}{(1-\gamma_2)^2}\right)}{(\xi_2+\frac{\gamma_2}{h})^2}\!\md\xi_2\\
&&= -\frac{(\beta_1 \gamma_1+\beta_2\gamma_2)(1-\gamma_2)^2}{4\pi^2h^2}\cdot \frac{2\pi ih}{(1-\gamma_2)^2}\cdot \frac{2\pi ih}{(1-\gamma_2)^2}\\
&&=\frac{\beta_1 \gamma_1+\beta_2\gamma_2}{(1-\gamma_2)^2}
\end{eqnarray*}
 
 Finally, the covariance is 
 \[\upsilon (f)=\frac{\kappa h^2}{ (1-\gamma_2)^4}+\frac{\beta_1 \gamma_1+\beta_2\gamma_2}{(1-\gamma_2)^2}.\]



\end{document}